\documentclass[preprint,11pt]{elsarticle}

\usepackage{amssymb}
 \usepackage{amsthm}
 \usepackage[symbol]{footmisc}
 \usepackage[colorinlistoftodos,prependcaption,textsize=small]{todonotes}

\newtheorem{lemma}{Lemma}[section]
\newtheorem{proposition}{Proposition}[section]

\newtheorem{theorem}{Theorem}[section]

\newcommand{\ToT}{\mbox{\scriptsize ToT}}

\journal{DAMath}

\begin{document}

\begin{frontmatter}



\title{The Sackin Index of Simplex 
%and Tree-child 
Networks}


\author[inst1]{Louxin Zhang}

\affiliation[inst1]{organization={Department of Mathematics and Centre of Data Science and Machine Learning, National University of Singapore},
addressline={10 Lower Kent Ridge Road},
           % city={Singapore},
           % postcode={119076}, 
           % state={Singapore},
           country={Singapore 119076}
          }

%\author[inst2]{Author Two}
%\author[inst1,inst2]{Author Three}

%\affiliation[inst2]{organization={Department Two},%Department and Organization
   %         addressline={Address Two}, 
   %         city={City Two},
   %         postcode={22222}, 
   %         state={State Two},
      %      country={Country Two}}

\begin{abstract}
A phylogenetic network is a simplex (or 1-component tree-child) network if the child of every reticulation node is a  network leaf. Simplex networks are  a superclass of phylogenetic trees and a subclass of tree-child networks. Generalizing  the Sackin index to phylogenetic networks, we prove that 
%\begin{itemize}
  %  \item 
the expected Sacking index of a random  simplex network is asymptotically $\Omega (n^{7/4})$ in the uniform model. %where networks are sampled uniformly random.
  %  \item The expected depth of a leaf is   asymptotically $\Omega (n)$ in a random tree-child network sample;
  %   \item The expected height of a random tree-child network is  asymptotically $\Omega (n)$. 
%\end{itemize}
\end{abstract}


\end{frontmatter}


\section{Introduction}
\label{sec:introduction}

Phylogenetic networks have been frequently used for modeling evolutionary history of genomes and genetic flow in population genetics. Since network models are much more complex than phylgoenetic trees,   different classes of  phylgoenetic networks have been introduced to investigate different issues of reconstruction of phylogenetic networks \cite{Cardona2009, galled_trees, galled_net,Francis2015}.  For these special classes of phylogenetic networks, algorithmic problems for determining  the relationship between phylogenetic trees and networks and for reconstruction of phylogenetic networks from DNA sequences, gene trees and other data  have been extensively studied \cite{gusfield2014book, huson2010book,gunawan2017IC}. 
%However,  existing algorithms for these problems have often high time complexity. 

The combinatorial and stochastic properties of different classes of phylogenetic networks has received increasing attention in the study of phylogenetic networks recently.   
Counting tree-child networks was first studied in \cite{McDiarmid2015}. 
A tight asymptotic value of the number of tree-child networks is  given in \cite{fuchs2021}.   Although  algorithms are  presented for enumerating tree-child networks \cite{cardona2019generation,zhang2019}, closed formulas and even simple recurrence  formulas for counting  tree-child networks  are unknown \cite{cardona2020,pons2021}.
Counting ranked tree-child networks is studied in 
\cite{bienvenu2020}. 
In addition, asymptotic and exact counts of galled trees and galled networks are given in \cite{bouvel2020counting} and \cite{gunawan2020,fuchs2022}, respectively. 

The expected height of random binary trees has been known for decades \cite{flajolet1982average}.
 Recently, the problem of computing  the  height of random phylogenetic networks is raised in \cite{McDiarmid2015,fuchs2022,stufler2021}.  
The Sackin index of a phylogenetic tree is defined to be the sum of the depths of its leaves \cite{Sackin1972,shao1990tree}. It is one of the widely-used indices used for measuring the balance of phylogenetic trees and testing evolutionary models \cite{avino2019tree,shao1990tree,xue2020scale}. In this paper, we will prove that  the expected  Sackin index of a random simplex network  on $n$ taxa (which are networks such that the child of each reticulation node is a leaf)  is $\Omega (n^{7/4})$ in the uniform model, which is significantly larger than the Sackin index of phylogenetic trees on $n$ taxa \cite{steel2016phylogeny}. 

The rest of this paper is divided into three parts. 
In Section~\ref{sec2}, basic concepts and notation of phylogenetic networks are introduced. In particular, we define the depth of nodes and the Sackin index of phylogenetic networks. In Section~\ref{sec3}, we present the bound of the expected Sacking index of a random simplex network in the uniform model that is mentioned above. In Section~\ref{sec4},
we conclude the study with several remarks and open questions.

\section{Basic concepts and notation}
\label{sec2}

\subsection{Tree-child networks}

For convenience,  we consider ``planted" phylogenetic networks over taxa (Figure~\ref{Fig1_example}). Such phylogenetic networks over a set $X$ of $n$ taxa are acyclic rooted graphs in which (i) the {\it root} is of out-degree 1, (ii) there are $n$ nodes of indegree 1 and outdegree 0, called the {\it leaves},   that are labeled one-to-one by the taxa of $X$, and (ii) all the other  nodes are of degree 3.

Each degree-3 node is called a {\it tree node}  if it is of out-degree 2 and indegree-1; it is a {\it reticulate node}  if it is of indegree 2 and out-degree 1. Note that binary  phylogenetic trees are simply  phylogenetic networks with no reticulate nodes. An edge $(p, q)$ is a {\it tree edge} if $q$ is either a tree node or a leaf; it is a {\it reticulation edge} if $q$ is a reticulation node.  

Let $N$ be a phylogenetic network.
We use $\rho$ to denote the root of $N$, ${\cal E}(N)$ to denote the set of edges. We also use  ${\cal L}(N)$,  ${\cal R}(N)$ and ${\cal T}(N)$ to denote the set of the leaves,  reticulation and tree  nodes, respectively. 
%The number of the nodes of $N$ is written as $|N|$.

Let $u, v\in \{\rho\}\cup {\cal L}(N)\cup {\cal R}(N)\cup {\cal T}(N)$.  If $(u, v)\in {\cal E}(N)$, $u$ is said to be a {\it parent} of $v$ and,  alternatively,  $v$ is a {\it child} of $u$). If there is a path from the network root to $v$ that passes $u$, $u$ is said to be an {\it ancestor} of $v$ and,  alternatively,  $v$ is a {\it descendant} of $u$. 

Let 
$e'=(p, q) \in {\cal E}(N)$ and $ e''=(s, t)\in {\cal E}(N)$. The edge 
$e'$  is said to be {\it above} $e''$  if $q$ is an ancestor of $s$, denoted by $e'\prec e''$. The edges $e'$ and $e''$  are said to be {\it parallel}, denoted by $e' \| e''$,  if neither of $e'$ and $e''$ is above the other.

A phylogenetic network is {\it simplex} if and only if the child of every reticulation is a leaf. The middle network in Figure~\ref{Fig1_example} is simplex,  where  the child of the reticulations are Leaf 3 and 4.

A phylogenetic network is {\it tree-child} if 
every non-leaf node has a child that is either a tree node or a leaf.  In Figure~\ref{Fig1_example}, the right phylogenetic network is a tree-child network with $2$ reticulation nodes. A binary phylogenetic network is tree-child if and only if for every  node $v$, there exists a leaf $\ell$ such that $v$ and $\ell$ are connected by a path consisting of tree edges. 

%A reticulation node is {\it abnormal} if one of its parents is an ancestor of another  and {\it normal} otherwise.  A tree-child network is {\it normal} if every reticulation node is normal.

It is easy to see that  phylogenetic trees are simplex networks, whereas  simplex networks are tree-child networks. Simplex networks are also called 1-component tree-child networks 
\cite{cardona2020}. 

%Figure 2
\begin{figure}[t!]
\centering
\includegraphics[scale=2]{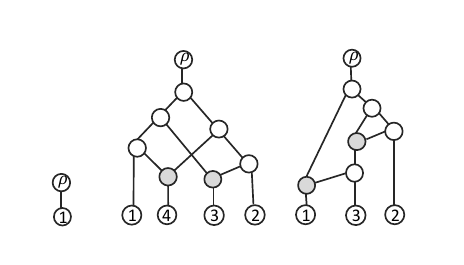}
\caption{ 
{A phylogenetic network with a single leaf (left), a simplex network on 4 taxa and a tree-child network on 3 taxa that is not simplex (right), where reticulation nodes are represented by filled circles.}   
}
\label{Fig1_example} 
\end{figure}

In this paper, we shall  use the following facts frequently without mention of them. 

\begin{theorem}
 Let $r(N)$, $\ell(N)$ denote the number of the reticulation nodes and leaves of a tree-child network $N$, respectively. Then, 
 \begin{enumerate}
 \item[{\rm (1)}] $|{\cal T}(N)|= \ell(N)+r(N)-1$.

 \item[{\rm (2)}]
 There are exactly $2\ell(N)+r(N)-1$ tree edges and $2r(N)$ reticulation edges in $N$. 
 \end{enumerate}
\end{theorem}

\subsection{Node depth, network height and Sackin index}

Let $T$ be a phylogenetic tree and $u\in {\cal V}(T)\cup {\cal L}(T)$. The {\it depth} of $u$ is defined to be the number of edges in the unique path from the tree root $\rho$ to $u$, which is also  equal to the number of ancestors of $u$.  Obviously, the depth of the tree root is 0.  The {\it height} of $T$ is defined to be the largest depth of a leaf in $T$.

In a phylogenetic network $N$, there are more than one path from the root $\rho$ to a specific descendant of a reticulation node.  We generalize the concept of node depth to phylogenetic networks as follows.  

Let $u$ be a node of $N$. The {\it depth} of a node $u$ is defined to be the number of edges in the longest path from the root $\rho$ to $u$, written as $d_N(u)$. The {\it ancestor number} of a node  $u$ is defined to be the number of the ancestors of the node, written as $\alpha_{N}(u)$. 
For example,  the depth and ancestral number of Leaf 4 are five are six, respectively,  in the right phylogenetic network in Figure~\ref{Fig2_example}. 

For a tree-child network $N$, we define the following parameters:
\begin{itemize}
    \item  The {\it height}  of $N$ is defined to be the largest depth of a leaf, denoted by $h(N)$. 
    
     \item  The {\it Sackin index} of $N$ is defined to be the sum of the depths of its leaves, denoted by $K(N)$.
\end{itemize}

Given a family ${\cal F}$ of tree-child networks. The {\it expected height} of a network of the family ${\cal F}$ in the uniform model is defined by: 
$$ \overline{H}({\cal F})= \frac{1}{|{\cal F}|} \sum_{N\in {\cal F}} h(N),$$
where $|{\cal F}|$ is the cardinality of the set $\cal F$.

The {\it expected Sackin index} of a  random network in ${\cal F}$ in the  is defined by:
$$ \overline{K}({\cal F})= \frac{1}{|{\cal F}|} \sum_{N\in {\cal F}} K(N).$$

\begin{theorem}
  Under the uniform model,
 \begin{enumerate}
   \item[{\rm (1)}] (\cite{flajolet1982average}) The expected height of a random phylogenetic tree over $n$ taxa is asymptotically $2\sqrt{\pi n}$.
   
   \item[{\rm (2)}] The expected Sackin index \footnote{The expected Sackin index of a phylogenetic tree on $n$ is different from that reported in literature by $n$ (see \cite{king2021simple,mir2013new} for example). The reason is that we work on "planted" phylogenetic trees.} of a random phylogenetic tree  on $n$ taxa is $\frac{2^{2n-2}n!(n-1)!}{(2n-2)!}$,
   whose  is asymptotically $\sqrt{\pi} n^{3/2}$.
  % \item[{\rm (3)}] The expected total tree node height of a random phylogenetic tree over $n$ leaves is  $2^nn!-\frac{(2n-1)!}{2^{n-1}(n-1)!}$, 
 %  which is asymptotically $2\sqrt{\pi} n^{3/2}$.
   
 \end{enumerate}
\end{theorem}

\section{ The expected Sackin index of random simplex networks}
\label{sec3}

In this section, we will use an enumeration procedure and a simple counting formula for simplex networks that appear in \cite{cardona2020} to obtain the asymptotic  Sackin index of a simplex network in the uniform model. 

\subsection{Enumerating simplex networks} 

We first briefly introduce  a procedure for enumerating simplex networks appearing in \cite{cardona2020}.  Let ${\cal OC}_{n}$ denote the class of simplex networks on $n$ taxa and
$o_{n}=|{\cal OC}_{n}|$.

Let $N\in {\cal OC}_{n}$. $N$ contains  It may contain 0 to $n-1$ reticulations. Recall that  the child of each reticulation is a leaf. 
All the tree nodes and leaves that are not below any reticulations are connected to the root by tree edges, forming a connected subtree, which we call the {\it top tree component} and  denote  by $C(N)$ (see \cite{gunawan2017IC}). For instance, the top tree component of the simplex network in the middle of Figure 1 consists of Leaf 1, Leaf 4 and their ancestors,  including $\rho$. 

Let $[i, j]$ denote the integer set $\{i, i+1, \cdots, j\}$, where $0<i\leq j$.
For any nonempty $I\subseteq [1, n]$, ${\cal OC}_{I,n}$ denotes the subset of simplex networks in which there are $n-|I|$ reticulations whose child are labeled uniquely with the elements of $[1, n] \setminus  I$.
Clearly, when $I=[1, n]$, ${\cal OC}_{I, n}$ is just  the set of phylogenetic trees on $[1, n]$.

For simplicity,  we write ${\cal OC}_{k, j}$ for 
${\cal OC}_{[1, k], k+j}$ for $j>1$.
The networks of ${\cal OC}_{k, j+1}$ can be generated by attaching the two grandparents of Leaf $j+1$ to each network $N$ of ${\cal OC}_{k, j}$ in the following two ways \cite{cardona2020}:
\begin{enumerate}
\item For each tree edge $e=(u, v)$ of $C(N)$, where $u\not\in {\cal R}(N)$, subdivide $e$ into  $(u, p), (p, q)$ and $(q, v)$,  and add three new edges
 $(p, r), (q, r)$ and $(r, j+1)$, This insertion operation is shown on the left network in Figure~\ref{Fig2_example}, where $j+1=4$. 
\item For each pair of tree edges $e'=(u, v)$ and 
   $e''=(s, t)$ of $C(N)$, i.e. $u\not\in {\cal R}(N)$ and $s\not\in {\cal R}(N)$, subdivide $e'$ into $(u, p)$ and 
   $(p, v)$ and $e''$ into $(s, q)$ and $(q, t)$,
   and add three new edges $(p, r), (q, r)$ and $(r, j+1)$. This insertion operation is shown in the right  network in Figure~\ref{Fig2_example}. 
\end{enumerate}

\begin{figure}[t!]
\centering
\includegraphics[scale=2]{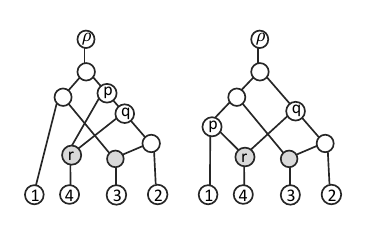}
\caption{ 
The Illustration of the procedure for enumerating simplex  networks by  attaching the grandparents ($p, q$) of a new leaf (4) below a reticulation node ($r$) onto either a tree edge (left) or a pair of tree edges (right) in the top tree component of a simplex network. 
}
\label{Fig2_example} 
\end{figure}

Each $N$ of ${\cal OC}_{k, j}$ contains exactly $2(k+j)-1$ tree edges in its $C(N)$. Additionally,  all the networks generated using the above method are distinct \cite{cardona2020}. This implies that 
$|{\cal OC}_{k, j+1}|=(2(k+j)-1)(k+j)
|{\cal OC}_{k, j}|$ and thus
%leads to the following simple formula for counting 1-component tree-child networks.
\begin{eqnarray}
 |{\cal OC}_{k, j}|= \frac{(2(k+j)-2)!}{2^{k+j-1}(k-1)!},
 \label{eqn1_zlx}
\end{eqnarray}
from which  the simple formula for counting the phylogenetic trees is obtained by setting $j$ to 0.

By symmetry,  ${\cal OC}_{ I, n}={\cal OC}_{J,n}$ if $|I|=|J|$. 
Therefore, we obtain \cite{cardona2020}: 
\begin{eqnarray}
 o_{n}= \sum^{n}_{k=1} {n \choose k} |{\cal OC}_{k, n-k}|
 %= \frac{(2n-2)!}{2^{n-1}}\sum^{n-1}_{k=0} {n \choose k} \frac{1}{(n-k-1)!}
  =\frac{(2n-2)!}{2^{n-1}}\sum^{n}_{k=0} {n \choose k} \frac{1}{(k-1)!}
  \label{total_count}
\end{eqnarray}
%Using Eqn.~(\ref{total_count}), we %obtain:
%$$o_1=1,\;\;o_2=3,\;\; o_3=39,\;\; o_4=1,095,\;\;o_5=52,605.$$

\subsection{A formula for the total depths of the nodes in the top tree component} 

Recall that $\alpha(u)$ denotes the number of ancestors of a node $u$. Let  
$\delta_{C(N)}(u)$ denote the number of descendants of $u$ that are in $C(N)$.  We also use  
$\mbox{ToT}(N)$ to denote the tree edges of $C(N)$.  We define:
 \begin{eqnarray}
   A_{C}(N)=\sum_{v\in C(N)} \alpha_{N}(v) =\sum_{u\in C(N)} \delta_{C(N)}(u) \label{eqn4_count}
 \end{eqnarray}
and
\begin{eqnarray}
 A_C({\cal OC}_{k, j})=
\sum_{N\in {\cal OC}_{k, j}}A_C(N).
\end{eqnarray}

%In the rest of this section, we will present a simple formula for $an_C(\mbox{1-}{\cal TC}_{k, j})$.

%For a phylogenetic network $N\in \mbox{1-}{\cal TC}_{k, j}$, we use $ToT(N)$ to denote the tree edges between two tree non-reticulation nodes of $N$. Additionally, for a non-reticulation node $u$ of $N$, we use $a_N(u)$ to denote  the set of ancestors of $u$ in $N$, For each tree node $u$,
%% \begin{eqnarray}
  %h_{leaf}(N)=\sum_{\ell\in {\cal L}(N)} a_{N}(\ell)=\sum_{u\in {\cal T}(N)} d'_{N}(u)\\
%   an_{leaf}(N)=\sum_{\ell\in {\cal L}(N)} a_{N}(\ell)=\sum_{u\in {\cal T}(N)} d'_{N}(u) \label{eqn2_count}\\
%   an_{tr}(N)=\sum_{v\in {\cal L}(N)\cup {\cal T}(N)} a_{N}(v)=\sum_{u\in {\cal T}(N)} d_{N}(u) \label{eqn3_count}
% \end{eqnarray}
% We further define the following parameter
% \begin{eqnarray}
%   an_{C}(N)=\sum_{v\in [1,n]\cup {\cal T}(N)} a_{N}(v) =\sum_{u\in {\cal T}(N)} d''_{N}(u) \label{eqn4_count}
% \end{eqnarray}
% where $d''_{N}(u)$ is the number of descendants  that are not Leaves from $k+1$ to $k+j$.
 
 \begin{lemma}
 Assume  $N \in {\cal OC}_{k, j}$ and let $n=k+j$. For each $e$ of $\mbox{\rm ToT}(N)$, we use $N(e)$ to denote the network obtained from $N$ by applying the first approach to $e$.
 Then, 
 \begin{eqnarray}
  \sum_{e\in ToT(N)} A_{C}(N(e))
  = (2n+3)A_{C}(N)+(2n-1) \label{eqn5_count}
  %
% \sum_{e\in ToT(N)} an_{tr}(N(e))
 % = (2n+1)an_{tr}(N)+3an_C(N)+(4n+2k-3) \label{eqn6_count}\\
  %
 % \sum_{e\in ToT(N)} an_{leaf}(N(e))
 % = (2n+1)an_{leaf}(N)+ an_{C}(N) 
 %   +2(n+k-1). \label{eqn7_count}
\end{eqnarray}
 \end{lemma}
 \begin{proof}
 Let $e=(u, v)\in \mbox{ToT}(N)$. By the description of the first approach, the set of nodes and edges of $N(e)$ are respectively: 
   $${\cal V}(N(e))={\cal V}(N)\cup \{p, q, r, n+1\}$$
   and  
   $${\cal E}(N(e))={\cal E}(N)\cup \{(u, p), (p, q), (q, v), (p, r), (q, r), (r, n+1)\} \setminus \{(u, v)\}$$
   (see Figure~\ref{Fig2_example}).
   
   Let ${\cal D}(v)$ be the set of descendants of $v$ in $C(N)$.
   We have the following facts: 
   \begin{eqnarray*}
    |\alpha_{N(e)}(x)|=2+|\alpha_{N}(x)|, && x\in {\cal D}(v)\\
    |\alpha_{N(e)}(y)|=|\alpha_{N}(y)|, &&  y\not\in {\cal D}_N(v)\cup\{v, p, q\}\\
    |\alpha_{N(e)}(p)|=|\alpha_{N}(v)|,&&\\
     |\alpha_{N(e)}(q)|=|\alpha_{N}(v)|+1,&&\\
     |\alpha_{N(e)}(v)|=|\alpha_{N}(v)|+2.&&
    % |a_{N(e)}(j+1)|=|a_{N}(v)|+2 .
     \end{eqnarray*}
Bu summing above equations, we obtain:
\begin{eqnarray}
 A_C(N(e))=A_{C}(N)+2\delta_{C(N)}(v)+ 2\alpha_{N}(v)+3.  \label{eqn8_count}
 %an_{tr}(N(e))=an_{tr}(N)+2d_{N}(v) 
 %  + 3a_{N}(v)+5. \label{eqn9_count}
\end{eqnarray}
Eqn.~(\ref{eqn4_count}) and (\ref{eqn8_count}) imply  that 
\begin{eqnarray*}
  && \sum_{e\in \ToT(N)} A_{C}(N(e))\\
  &=&  (2n-1)A_{C}(N) +  2\sum_{(u, v)\in \ToT(N)} \delta_{C(N)}(v)+ 2\sum_{(u, v)\in \ToT(N)} \alpha_{N}(v) + 3\\
  &=&(2n-1)A_{C}(N)+2\left(\sum_{ v\in C(N)} \delta_{C(N)}(v) - \delta_{C(N)}(\rho)\right) +2A_{C}(N)  + 3(2n-1)\\
   &=& (2n+3) A_{C}(N) + (2n-1),
\end{eqnarray*}
where $n=k+j$ and we use the fact that 
$\delta_{C(N)}(\rho)=2n-1$ for the network root $\rho$.

%From Eqn.~(\ref{eqn9_count}), we obtain: 
%\begin{eqnarray*}
% \sum_{e\in ToT(N)} an_{tr}(N(e))
%  &=& (2(k+j)-1)an_{tr}(N)+2\sum_{(u, v)\in ToT(N)} d_{N}(v) \\
%  && +3\sum_{(u, v)\in ToT(N)} a_{N}(v)  + 5(2(k+j)-1)\\
%  &=& (2(k+j)-1)an_{tr}(N)+2\left(\sum_{ v\in {\cal T}(N)} d_{N}(v) - d_{N}(\rho)\right)  \\
%   && +3an_{C}(N)  + 5(2(k+j)-1).\\
%\end{eqnarray*}
%By the fact that $d_{N}(\rho)=2k+3j-1$ and Eqn.(\ref{eqn2_count}) and (\ref{eqn3_count}), we obtain 
%Eqn. (\ref{eqn6_count}). 

%\newpage
%Similarly, for each leaf descendant $x$ of $v$, $a_{N(e)}(x)=2+a_{N}(x)$ if $v$ is not a leaf; and  $a_{N(e)}(v)=2+a_{N}(v)$ if $v$ is a leaf.
%   For any other leaf $y\not\in d'_N(v)\cup\{ j+1\}$, 
 %  $a_{N(e)}(y)=a_{N}(y)$.
 % Additionally, 
 %   $a_{N(e)}(j+1)=a_{N}(v)+2.$
 %   Thus, 
  %  $$an_{\lf}(N(e))= an_{\lf}(N)+2d^{\lf}_{N}(v)
  %   + an_{N}(v)+2 + \delta_v,$$
 %   where $\delta_v=2$ if $v$ is a leaf.
    
%Summing over all edges of $\ToT(N)$ and $n=k+j$, we obtain: 
%  \begin{eqnarray*}
% && \sum_{e\in \ToT(N)} an_{\lf}(N(e))\\
%  &=& (n-1)an_{\lf}(N)+2\sum_{(u, v)\in \ToT(N)} d^{\lf}_{N}(v) +2k \\
 % && + \sum_{(u, v)\in \ToT(N)} (an_{N}(v)  + 2)\\
 %  &=& (2n-1)an_{leaf}(N)+2\left(\sum_{ v\in {\cal T}(N)} d^{\lf}_{N}(v) -d^{\lf}_{N}(\rho)\right)+2k\\
 %  && +\sum_{(u, v)\in ToT(N)} (an_{N}(v)  + 2)\\
 %  &=& (2n+1)an_{leaf}(N)+ an_{C}(N) +2(n+k-1).
%\end{eqnarray*}

 \end{proof}
 
 %For each pair of tree edges $e', e''$ such that $e'\in ToT(N), e''\in ToT(N)$, 
% $e'$ is an ancestral edge of $e''$, 
% $e''$ is an ancestral edge of $e'$,
% or $e'$ and $e''$ are parallel.

\begin{lemma}
 Let  $N\in {\cal OC}_{k, j}$ and $N(e',e'')$ denote the network obtained from $N$ using the second approach to a pair of distinct edges  $e', e''$ of $C(N)$.
 Then, 
 {\small 
 \begin{eqnarray}
  \sum_{\{e',e''\}\subset ToT(N)} A_{C}(N(e', e''))
  =  (2n^2+n-2) A_C(N) -(2n-1), \label{eqn10_count}
\end{eqnarray}
}
where $n=k+j$.
 \end{lemma}
 \begin{proof}
 After attaching the reticulation node  $r$ onto the edges  
 $e'=(u, v)$ and $e''=(s, t)$, the tree edge $e'$ is divided into $e'_1=(u, p)$ and $e''_1=(p, v)$;   the edge $e''$ is divided into $e'_2=(s, q)$ and $e''_2=(q, t)$, where $p, q$ are the parents of $r$ in $N(e', e'')$. 
 We consider  two possible cases. 
 
 First, we consider the case that $e'\prec e''$.
 In $N(e', e'')$, for each descendant $x$  of $u$ in $C(N(e', e'')$  such that $x\neq p$ and $x$ is not below  
 $s$,
  $\alpha(x)$ increases by 1 because of the subdivision of $e'$. 
 For each descendant $y$ of $s$ in $C(N(e', e'')$ such that   $y\neq q$, 
  $\alpha(y)$ increases by 2 because of the subdivision of $e'$ and $e''$.
Additionally, 
  $$\alpha_{N(e', e'')}(p)=\alpha_{N}(v)$$
  and because $q$ is below $v$,
  $$ \alpha_{N(e', e'')}(q)=\alpha_N(t)+1.$$
 For any tree node or leaves that are not the descendants of $u$, $\alpha (x)$ remains the same. Summing all together, 
 we have that 
 \begin{eqnarray}
   A_{C}(N(e', e'')) 
   = A_C(N) + \delta_{C(N)}(v) + \delta_{C(N)}(q)  + \alpha_{N}(v)
     + \alpha_N(q)+3,
     \label{eqn10_1}
 \end{eqnarray}
 where 3 is the sum of  the total increase of 
 $\alpha(v)$, $\alpha(q)$ and $\alpha(t)$.

 Summing over all the comparable edge pairs, we have
 \begin{eqnarray}
  &&\sum_{e', e''\in \ToT(N): e'\prec e''} A_{C}(N(e', e'')) \nonumber \\
  &=& (A_C(N)+3)|\{(e', e'') \;|\; e'\prec e''\}| + \sum_{(u, v)\in \ToT(N)} (\delta_N(v)+\alpha_N(v))\delta_{C(N)}(v)  \nonumber \\
  && + \sum_{(u, v)\in \ToT(N)} (\delta_{C(N)}(v)+\alpha_N(v))(\alpha_N(v)-1) \nonumber \\
   &=& (A_C(N)+3)|\{(e', e'')\;|\; e'\prec e''\}| 
   + \sum_{v\in C(N)\setminus\{\rho\}} (\delta_{C(N)}(v)+\alpha_N(v))^2 \nonumber \\
  && 
   - \sum_{v\in C(N)\setminus\{\rho\}} (\delta_{C(N)}(v)+\alpha_N(v)) \nonumber \\
  &=& (A_C(N)+3)|\{(e', e'')\;|\; e'\prec e''\}| 
   + \sum_{v\in C(N)\setminus\{\rho\}} (\delta_{C(N)}(v)+\alpha_N(v))^2 \nonumber \\
  && 
   - 2A_C(N)+(2n-1),
   \label{eqn12_count}
 \end{eqnarray}
 where we use the fact that
 $\sum_{v\in C(N)\setminus\{\rho\}} \delta_N(v)=A_C(N)-(2n-1)$.

 Second, we assume  $e'\| e''$.  
 In $N(e', e'')$, for each descendant $x$  of $u$ in $C(N(e', e''))$,
  $\alpha(x)$ increases by 1 because of the subdivision of $e'$. 
 For each descendant $y$ of $p$ in $C(N(e', e''))$ , 
  $\alpha(y)$ increases by 1  because of the subdivision of $e''$.
Additionally, 
  $\alpha_{N(e', e'')}(p)=\alpha_{N}(v)$
  and 
  $ \alpha_{N(e', e'')}(q)=\alpha_N(t).$
 For any tree node or leaves that are the descendants of neither $u$ nor $p$, $\alpha(x)$ remains the same.
 Hence, 
 \begin{eqnarray}
  A_C(N(e', e''))=A_C(N)+\delta_{C(N)}(v)+\delta_{C(N)}(q)+\alpha_{N}(v)+\alpha_N(q)+2, \label{eqn10_2}
 \end{eqnarray} 
 where $2$ counts for the increase by 1 of both $\alpha(v)$ and $\alpha(t)$.

 Summing Eqn.~(\ref{eqn10_2}) over parallel edge pairs, we obtain:
  \begin{eqnarray}
 &&\sum_{e', e''\in \ToT(N): e'\| e''} A_{C}(N(e', e''))  \nonumber \\
 &=&  (A_C(N)+2) \left|\{(e', e''):\; e'\|e''\}\right| \nonumber \\
  && + \sum_{(u, v)\in \ToT(N)} 
   o_N((u, v)) (\delta_{C(N)}(v)+\alpha_N(v)),  
   \label{eqn13_count}
 \end{eqnarray}
 where $o_{N}((u, v))$ is the number of the edges of $C(N)$ that are parallel to $(u, v)$.

 Summing Eqn.~(\ref{eqn12_count}) and (\ref{eqn13_count}) and using the facts that $n=k+j$ and $|\{(e', e'')\;|\; e'\prec e''\}|=A_C(N)-(2n-1)$ and 
 $$\sum_{(u, v)\in \ToT(N)} 
    (\delta_{C(N)}(v)+\alpha_N(v))=2A_C(N) - \delta_{C(N)}(\rho)=2A_C(N)-(2n-1),$$ we obtain:
  \begin{eqnarray}
 && \sum_{\{e', e''\}\subset ToT(N)} A_{C}(N(e', e'')) \nonumber  \\
 &=& (A_C(N)+2) {2n-1\choose 2} + 
 |\{(e', e'')\;|\; e'\prec e''\}|  -2A_C(N)+(2n-1) \nonumber \\
 &&
   + (2n-1) \sum_{(u, v)\in \ToT(N)} 
    (\delta_{C(N)}(v)+\alpha_N(v)) \nonumber \\
  &=& (A_C(N)+2) {2n-1\choose 2} - A_C(N) 
   + (2n-1) (2A_C(N) - (2n-1)) \nonumber\\
 % &=& (an_C(N)+2) {2n-1\choose 2} + 
 %(an_C(N)-(2k+2j-1)) \nonumber \\
 %&& -2an_C(N)+(2k+2j-1) \nonumber \\
% &&
%   + (2k+2j-1) (2an_C(N) - (2k+2j-1)) \nonumber \\ 
   &=& A_C(N) (2n^2+n-2)  -(2n-1). \nonumber
  \end{eqnarray}
 \end{proof}

\begin{theorem}
Let $k\geq 1$ and $j\geq 0$.
  For the subclass of simplex networks with $j$ reticulations whose children are $k+1, \cdots, k+j$ on $n=k+j$ taxa,  
  \begin{eqnarray}
  A_C({\cal OC}_{k, j}) =\frac{(2n)!}{2^j(2k)!}  \left( 
  2^kk!- \frac{(2k-1)!}{2^{k-1}(k-1)!} \right).
  \label{eqnXXX_count}
  \end{eqnarray}
\end{theorem}
\begin{proof}
  Summing Eqn.~(\ref{eqn5_count}) and (\ref{eqn10_count}), we obtain
  \begin{eqnarray*}
  \sum_{e\in  ToT(N)} A_C(N(e))+ \sum_{\{e',e''\}\subset ToT(N)} A_{C}(N(e', e''))
  =  (2n+1)(n+1) A_C(N) 
  \end{eqnarray*}
  for each $N\in {\cal OC}_{k, j}$. 
  Summing the above equation over all networks of ${\cal TC}_{k, j}$, we have
  \begin{eqnarray*}
  A_C({\cal OC}_{k, j+1})
  = (2n+1)(n+1)\sum _{N\in {\cal OC}_{k, j}} A_C(N)
 = A_C({\cal OC}_{k, j})
  \end{eqnarray*}
 or,  equivalently, 
 $
  A_C({\cal OC}_{k, j})
  = n(2n-1)A_C({\cal OC}_{k, j-1}).
  $
  Since it is proved in \cite{zhang2019} that
  \begin{eqnarray*}
 A_C({\cal OC}_{k, 0})= 2^kk!- \frac{(2k-1)!}{2^{k-1}(k-1)!}, 
 \label{eqnxxx20_count}
  \end{eqnarray*}
  we obtain Eqn.~(\ref{eqnXXX_count}) by induction.
\end{proof}
 
 %\noindent {\bf Remark}  The Sac index of a phylogenetic tree over $n$ taxa is defined to measure the structural balance  \cite{Sackin1972}. It is actually equal to $A_C(T)-(2n-1)$. 
%The Eqn.~(\ref{eqnxxx20_count}) for the class ${\cal OC}_{n, 0}$
% of phylogenetic trees over $n$ taxa was independently reported in \cite{zhang2019} and \cite{king2021simple}. Hence,
% Eqn.~(\ref{eqnXXX_count}) is a generalization of their result.

 \subsection{The expected total c-depth of random simplex networks}
 %\todo{revise title}
 
 The {\it expected total c-depth} of  simplex networks with $j$ reticulations on $k+j$ leaves is defined as: 
 \begin{eqnarray}
 D_{C}(k, j)=\frac{ {k+j \choose k} \sum_{N\in {\cal OC}_{k,j}}A_{C}(N)}{{k+j\choose k} |{\cal OC}_{k,j}|} =\frac{ A_{C}({\cal OC}_{k,j})}{ |{\cal OC}_{k,j}|}
 \end{eqnarray}
 in the uniform model.
 For each $n\geq 1$,  the expected total c-depth of simplex networks on $n$ taxa becomes:
 \begin{eqnarray}
 \overline{D}_C(n) =
 %\left( \sum^{n}_{k=1}{n \choose k} \sum_{N\in \mbox{1-}{\cal TC}_{k,n-k}}an_{C}(N) \right)
 \left( \sum^{n}_{k=1}{n \choose k} 
  A_{C}({\cal OC}_{k,n-k})
 \right)
 \mathbin{/}
 \left( \sum^{n}_{k=1}{n\choose k} |{\cal OC}_{k,n-k}|\right)
 \end{eqnarray}
 in the uniform model.

 \begin{proposition}
   For any $k\geq 1$ and $j\geq 0$, 
   the average total component depth  
   $D_C(k, j)$ has the following asymptotic value. 
   \begin{eqnarray}
    D_C(k, j)=\frac{\sqrt{\pi}n(2n-1)}{\sqrt{k}} \left(1 - \pi^{-1/2}k^{-1/2} + O(k^{-1})\right), \label{eqn16_mean}
   \end{eqnarray}
   where $n=k+j$.
 \end{proposition}
 \begin{proof}
  Since ${2k \choose k} = 
  \frac{2^{2k}}{\sqrt{\pi k}}(1-(8k)^{-1}+O(k^{-2})),$ by
  Eqn.~(\ref{eqn1_zlx}),
  \begin{eqnarray*}
   D_C(k, j)&=&
     \frac{(2n)!}{2^{n-k}(2k)!}
    \left( 
  2^kk!- \frac{(2k-1)!}{2^{k-1}(k-1)!} \right)/\left(
  \frac{(2n-2)!}{2^{n-1}(k-1)!}\right)\\
  &=&  \frac{n(2n-1)}{k}
    \left( 
  \frac{2^{2k}k!k!}{(2k)!}- 1 \right)\\
  &=&  \frac{n(2n-1)}{k}\left(\frac{\sqrt{\pi k}} {1-(8k)^{-1}+O(k^{-2})}-1 \right)\\
  &=& \frac{\sqrt{\pi}n(2n-1)}{\sqrt{k}} \left(1-\frac{1}{\sqrt{\pi k}}+\frac{1}{8k} + O(k^{-2})\right).
  \end{eqnarray*} 
  Thus, Eqn.~(\ref{eqn16_mean}) is proved.
  \end{proof}
  
  %\newpage

  \begin{proposition}
   For any $n>1$, 
    $\overline{D}_C(n)$ has the following bounds.
    {\small 
   \begin{eqnarray}
    (\sqrt{2}-1)
       (1+O(n^{-1/4})) \leq \frac{\overline{D}_C(n)}{\sqrt{\pi}n^{3/4}(2n-1)}
       \leq 2
       (1+O(n^{-1/4})).
       \label{eqn16_prop32}
   \end{eqnarray}
   }
 \end{proposition}
  
  \begin{proof}
  Let $ S_{k}= {n \choose k}|{\cal OC}_{k, n-k}|=  {n\choose k} \frac{(2n-2)!}{2^{n-1}(k-1)!} $ and $k_0=\sqrt{n+1}-1$. 
  
 By considering the ratio $S_{k+1}/S_{k}$, we can show (in Appendix A):
 \begin{itemize}
     \item $S_{k}\leq S_{k+1}$ for $k\in [1, k_0]$;
     \item  $S_{k} > S_{k+1}$ if $k_0<k<2k_0$;  
  \item  
  $S_{k}>2S_{k+1}$ if $k\geq 2k_0$.
  \end{itemize}

  First, we define  $$
  f_1=\left(\sum_{k\leq k_0} D_C(k, n-k) S_k\right) \mathbin{/} \left(\sum_{k\leq k_0}S_k\right).$$
  Since 
  $D_C(k, n-k)$ is a decreasing function (which is proved in Appendix A), 
  \begin{eqnarray*}
  f_1
  \geq D_C({\lfloor k_0\rfloor}, n-{\lfloor k_0\rfloor})=\sqrt{\pi}n^{3/4}(2n-1)(1+O(n^{-1/4}))
  \end{eqnarray*}
  Furthermore, 
  since $S_k$ is increasing on $[1, k_0)$, 
  \begin{eqnarray*}
  f_1
    &\leq & \frac{1}{\lfloor k_0\rfloor}\sum_{k \leq \lfloor k_0\rfloor}D_C(k, n-k)\\
    &=& \frac{\sqrt{\pi}n(2n-1)}{\lfloor k_0\rfloor}\sum_{k\leq \lfloor k_0\rfloor}
     \left(k^{-1/2}-\sqrt{\pi} k^{-1} 
     + (1/8) k^{-3/2} + O(k^{-5/2}\right)\\
     &=&  \frac{\sqrt{\pi}n(2n-1)}{\lfloor k_0\rfloor} 
     \left[ 2(\lfloor k_0\rfloor^{1/2}-1) -\sqrt{\pi} \ln \lfloor k_0\rfloor   + O(\lfloor k_0\rfloor^{-1/2})\right]\\
     &=& 2\sqrt{\pi}n^{3/4}(2n-1) (1+ O(n^{-1/4}))
  \end{eqnarray*}

 Second, define
  \begin{eqnarray}
  f'=\frac{\sum_{k\in [k_0, 2k_0]} D_C(k, n-k) S_k}{\sum_{k\in [k_0, 2k_0]}S_k}.
  \label{eqn17_mean}
  \end{eqnarray}
 Since the
   $D_C(k, n-k)$ is a decreasing function  and $S_k$ is  decreasing on $[k_0, 2k_0]$,   
  $$f'\leq D_C(\lceil k_0\rceil, n-\lfloor k_0\rfloor) = \sqrt{\pi}n^{3/4}(2n-1)(1+O(n^{-1/4}))$$  and, by Eqn.~(\ref{eqn16_mean}), 
  \begin{eqnarray*} 
  f'&\geq & \frac{1}{k_0} \sum_{k\in [k_0, 2k_0]} D_C(k, n-k)\\
  &=& \frac{\sqrt{\pi}n(2n-1)}{k_0} \sum_{k\in[k_0, 2k_0]} \left(k^{-1/2}-\sqrt{\pi} k^{-1} 
     + (1/8) k^{-3/2} + O(k^{-5/2}\right)\\
    &=& \frac{\sqrt{\pi}n(2n-1)}{k_0} \left[ 2 ( \sqrt{2}-1) k_0^{1/2} 
    - \sqrt{\pi} \ln 2 + 
      (1/4) (1-\sqrt{1/2}) k_0^{-1/2} + O(k_0^{-1})\right]\\
      &=& 2(\sqrt{2}-1)\sqrt{\pi}n^{3/4}(2n-1)
       (1+O(n^{-1/4}))
  \end{eqnarray*} 
  
  Now, we consider 
  $$ f_2=\frac{\sum_{k > k_0}D_C(k, n-k)S_{k}
    }{\sum _{k>k_0}S_k}.$$
  For each $k> 2k_0$, $D_C(k, n-k) \leq D_C(2k_0, n-2k_0) \leq f'$.  This implies that  
  $$\sum_{k>2k_0}D_C(k, n-k)S_k \leq f'
  \sum_{k>2k_0}S_k$$ and thus  
  \begin{eqnarray*}
    f_2 & =&  \frac{\sum_{k_0\leq k\leq 2k_0}D_C(k, n-k)S_{k}
   + \sum _{k>2k_0}D_C(k, n-k)S_k }{\sum _{k>k_0}S_k} \\
   &=& \frac{f'\sum_{k_0\leq k\leq 2k_0}S_{k}
   + \sum _{k>2k_0} D_C(k, n-k)S_k}{\sum _{k>k_0}S_k} \\
   &\leq& \frac{f'\sum_{k_0\leq k\leq 2k_0}S_{k}
   + f' \sum _{k>2k_0} S_k}{\sum _{k>k_0}S_k} \\
   &=& f'
  \end{eqnarray*}
  On the other hand, 
  $S_{k}\geq 2S_{k+1}$ for $k\geq 2k_0$
 and thus $\sum _{k> 2k_0} S_k\leq 2S_{\lceil 2k_0\rceil}$
  \begin{eqnarray*}
    f_2
    & \geq &  \frac{\sum_{k\in [k_0, 2k_0]}D_C(k, n-k)S_{k}}{\sum _{k>k_0}S_k} 
   = \frac{f'\sum_{k\in [k_0, 2k_0]}S_{k}}{\sum _{k\in [k_0, 2k_0]}S_k +\sum _{k> 2k_0} S_k } \\
   &\geq& \frac{f'\sum_{k\in [k_0, 2k_0]}S_{k}}{2\sum _{k\in [k_0, 2k_0]}S_k} = f'/2
  \end{eqnarray*}

  The bounds on $f_1$ and $f_2$ implies that the mean value $f$ over the entire region is between $f_1$ and $f_2$, Thus we have proved that 
  $$ (\sqrt{2}-1)\sqrt{\pi}n^{3/4}(2n-1)
       (1+O(n^{-1/4}))\leq  \min (f'/2, f_1) \leq 
       \overline{D}_C(n),$$
       and $$ \overline{D}_C(n)\leq \max (f_1, f_2)
       \leq 2\sqrt{\pi}n^{3/4}(2n-1)
       (1+O(n^{-1/4})).$$
 \end{proof}
\subsection{Bounds on the Sackin index for random simplex network}

Recall that 
$K(N)=\sum_{\ell\in {\cal L}(N)}d(\ell)$
for a network $N$.

\begin{proposition}
\label{prop33}
 For a simplex network
 $N\in {\cal OC}_{k, j}$, 
 $$K(N) \leq A_C(N) +1 \leq 2 K(N)$$
 where $k\geq 1, j\geq 0$. 
\end{proposition}
\begin{proof}
 Let $N\in {\cal OC}_{k, j}$. 
  If $j=0$, $N$ does not contain any reticulation node and thus every node of $N$ is in the top tree component of $N$. By definition, 
  $K(N) \leq A_C(N)$. 
  Since $N$ is a phylogenetic tree,   $N$ contains  the same number of internal nodes (including the root $\rho$) as the number of leaves.  By induction, we can prove that there exists a 1-to-1 mapping $\phi: {\cal T}(N) \cup \{\rho\}\rightarrow {\cal L}(N)$ such that  $u$ is an ancestor of  $\phi(u)$ for every $u$ (see {Appendix B})
  Noting that $d_N(\rho)=0$ and $d_{N}(\phi(\rho))\geq 1$, we have:
    $$A_C(N) +1
    %=\sum_{\ell\in {\cal L}(N)} d_N(\ell) 
   %  + \sum_{\ell\in {\cal T}(N)} d_N(u)
     \leq \sum_{\ell\in {\cal L}(N)} d_N(\ell) 
     + \sum_{\ell\in {\cal T}(N)} d_N(\phi(u))=2K(N).$$
     
 We now generalize the above proof for phylogenetic trees to the general case where $j>0$ as follows.
  
  We assume that $r_1,r_2, \cdots, r_j$ are the $j$ reticulation nodes and their parents are $p'_j$ and $p''_j$. 
  Since $N$ is simplex, $p'_j$ and $p''_j$  are both found in the top tree components.
  Clearly, $d_{N}(p'_j)\geq 1, d_N(p''_j)\geq 1$. 
 Since Leaf $(k+j)$ is the child of $r_j$, 
 $d_N(k+j)=\max (d_{N}(p'_j), d_{N}(p''_j))+2 
   \leq d_{N}(p'_j) + d_{N}(p''_j)$ unless 
   $\min (d_{N}(p'_j), d_{N}(p''_j))=1$.

  If $d_{N}(p'_i)=\min (d_{N}(p'_i), d_{N}(p''_i))=1$,   $p'_i$ is the unique child of the root $\rho$. This implies that $N$ contains at most an $i$ for which $\min (d_{N}(p'_i), d_{N}(p''_i))=1$.
  Since $N$ is tree-child, 
  the parents $p'_i$ and $p''_i$ are distinct for different $i$, 
  $K(N)\leq A_C(N)+1.$

  Without loss of generality, we may further assume 
  $d'_N(p'_1)=1$.  The tree-component $C(N)$ contains 
  $k +2j$ internal tree nodes including $\rho$.  
  We set 
  $${\cal T}(C(N))\setminus \{p'_i, p''_i : 1\leq i\leq j\}=\{\rho, u_1, \cdots, u_{k-1}\}.$$
  Again, there is an 1-to-1 mapping $\phi$ from 
  $\{\rho, u_1, \cdots, u_{k-1}\}$ to ${\cal L}(C(N))=\{1, 2, \cdots, k\}$ such that the leaf $\phi(u_i)$ is a descendant of $u_i$. 
  Therefore, since $d_{N}(\rho)=0$ and $d_N(\phi(\rho))\geq 1$, 
  \begin{eqnarray*}
   d_N(\rho)+1 &\leq &  d_{N}(\phi(\rho)); \\
   d_{N}(p'_i)+d_{N}(p''_i)
  &\leq &  2\max(d_{N}(p'_i), d_{N}(p''_i)) < 2 d_N(k+i), 1\leq i\leq j;\\
  d_{N}(u_i)+d_{N}(\phi(u_i))&\leq& 2d_{N}(\phi(u_i)), 1\leq i\leq k-1.
  \end{eqnarray*}
  Therefore,
   $A_C(N)+1\leq  2K(N).$
\end{proof}

\begin{theorem}
 The expected Sackin index $\overline{K}({\cal OC}_{n})$ of a simplex network on $n$ taxa  is $\Theta(n^{7/4})$.
\end{theorem}
\begin{proof}
We define 
$$ K({\cal OC}_{k, j}) =\sum_{N\in {\cal OC}_{k, j}}K(N),$$
where $k\geq 1$ and $j\geq 0$. By Proposition~\ref{prop33}, 
  \begin{eqnarray*}
   K({\cal OC}_{k, j}) \leq A_{N}({\cal OC}_{k, j})+ |{\cal OC}_{k, j}| \leq 2K({\cal OC}_{k, j}).
   \end{eqnarray*}
   
   By Eqn.~(\ref{eqn16_prop32}), we have: 
   \begin{eqnarray*}
   \overline{K}({\cal OC}_{n}) &=&   \frac{1}{
     \sum^{n}_{k=1}{n\choose k} |{\cal OC}_{k, n-k}|}\left(\sum^{n}_{k=1}{n\choose k} K\left({\cal OC}_{k, n-k}\right)  \right)\\
     &\leq &  \frac{1}{
     \sum^{n}_{k=1}{n\choose k} |{\cal OC}_{k, n-k}|}\left(\sum^{n}_{k=1}{n\choose k} A_C\left({\cal OC}_{k, n-k}\right) \right) +1\\
     &=& \overline{D}_C(n)+1 \\
     &=& 4\sqrt{\pi}n^{7/4}+1.
  \end{eqnarray*}
  Similarly, using Eqn.~(\ref{eqn16_prop32}),  we have that 
  $$ 2\overline{K}({\cal OC}_{n}) \geq
     \overline{D}_C(n) = 2(\sqrt{2}-1)\sqrt{\pi}n^{7/4}+O(1),$$
    equivalently,
    $$\overline{K}({\cal OC}_{n}) \geq  (\sqrt{2}-1)\sqrt{\pi}n^{7/4}+O(1).$$
    This concludes the proof.
\end{proof}

\section{Conclusion}
\label{sec4}

%In this short paper, we have presented a recurrence formula for counting tree-child networks with $k$ reticulations over $n$ taxa, where $k<n$.  An interesting question for future study is whether such a formula can be used to study the following open questions. 

What facts about phylogenetic trees remain valid for phylogenetic networks is important in the study of phylogenetic networks. 
In this paper, an asymptotic estimate (up to constant ratio) for the expected Sackin index of a  simplex network is given in the uniform model. This study raises a few research problems. First,  the expected Sackin index of tree-child networks over $n$ taxa is still unknown.    It is also interesting to investigate the Sackin index for galled trees, galled networks and other  classes of networks
(see \cite{Zhang_18} for example).

Secondly, it is even more challenging to estimate the expected height of simplex networks and tree-child networks.
It is interesting to see whether or not the approach introduced by Stufler \cite{stufler2021} can be used to answer this question.

%What is the ratio of normal networks to tree-child networks over $n$ taxa? See \cite{McDiarmid2015}  for a partial answer.

 \bibliographystyle{elsarticle-num} 
 \bibliography{Refs_RecurrenceFormula}

\newpage

 \section*{Appendix}
 
 \subsection*{Appendix A}
 
\noindent {\bf Proposition A.1.}
  Let $ S_{k}=   {n\choose k} \frac{(2n-2)!}{2^{n-1}(k-1)!} $ and $k_0=\sqrt{n+1}-1$.  Then, 
 \begin{itemize}
     \item $S_{k}\leq S_{k+1}$ for $k\in [1, k_0]$;
     \item  $S_{k} > S_{k+1}$ if $k_0<k<2k_0$;  
  \item  
  $S_{k}>2S_{k+1}$ if $k\in [2k_0, \infty)$.
  \end{itemize}

 \begin{proof}
 Note that
   $$S_{k+1}=\frac{n-k}{k(k+1)}S_k.$$
   If $k\leq k_0$, $(k+1)^2\leq (\sqrt{n+1}-1+1)^2=n+1$ and, equivalently, 
   $k(1+k)\leq n-k$ and thus  ${S_{k+1}}=\frac{n-k}{k(k+1)}S_k\geq S_k$.
  Similarly, ${S_{k+1}}< S_k$ if $k>k_0$.
  
 If $k\geq 2k_0$, $(k+2)^2\geq 4(n+1)$ and $k^2+k\geq 4(n-k)+k\geq 4(n-k)$ and therefore, ${S_{k+1}}\leq \frac{1}{4} S_k< \frac{1}{2} S_k$.
 \end{proof}
 
 \noindent {\bf Proposition A.2.}
    $D_C(k, n-k)$ is a decreasing function of $k$ on 
    $[1, n-1]$.

 \begin{proof}
 By Eqn.~(\ref{eqn1_zlx}) and Eqn.~(\ref{eqnXXX_count}), 
    $$D_C(k, n-k)=\frac{n(2n-1)}{k}\left( \frac{4^kk!k!}{(2k)!}-1\right).
  $$
  Hence,
 \begin{eqnarray*}
   && n(2n-1)[D_C(k, n-k)-D(k+1, n-k-1)]\\
   &=&  \frac{2^{2k}(k-1)!k!}{(2k)!} - \frac{2^{2k}4k!(k+1)!}{(2k+2)!} - \frac{1}{k(k+1)}\\
   &=& \frac{2^{2k}k!(k-1)!}{(2k+1)!} - \frac{1}{k(k+1)}\\
   &=& \frac{1}{k(k+1)}\left[\frac{2^{2k}k!(k+1)!}{(2k+1)!} - 1\right]\\
   &=& \frac{1}{k(k+1)}\left[ \frac{  4\times 6\times \cdots 2k\times (2k+2)}{ 3\times 5\times  \cdots (2k+1)} -1\right]\\
  &>& 0.
 \end{eqnarray*}
  
 \end{proof}

\subsection*{Appendix B}

\noindent {\bf Proposition B.1.}
Let $P$ be a phylogenetic tree on $n$ taxa.
there exists a 1-to-1 mapping $\phi: {\cal T}(P)\cup \{\rho\}\rightarrow {\cal L}(P)$ such that  $u$ is an ancestor of  $\phi(u)$ for each $u\in {\cal T}(P) \cup \{\rho\}$.

\begin{proof}
  We prove the fact by mathematical induction on $n$. When $n=1$, we simply map the root 
  $\rho$ to the only leaf.  
  
 Assume the fact is true for $n\leq k$, where $k\geq 1$. For a phylogenetic tree $P$ with $k+1$ leaves, we let the child of the root $\rho$ be $u$ and the two grandchildren be $v$ and $w$. We consider the subtree $P'$ induced by  $u$, $v$ and all the descendants of $v$ and the subtree $P''$ induced by $u$, $w$ and all the descendants of $w$.  
 
 Obviously, both $T'$ and $T''$ have less than $k$ leaves. By induction, there is a 1-to-1 mapping $\phi': {\cal T}(P')\cup \{u\} \rightarrow {\cal L}(P')$ satisfying the constraints on leaves, and
 there is a 1-to-1 mapping $\phi'': {\cal T}(P'')\cup \{u\} \rightarrow {\cal L}(P'')$ satisfying the constraints on leaves.
 Let $\phi'(u)=\ell$. Then, the function that maps $\rho$ to $\ell$, $u$ to $\phi''(u)$ and all the other tree nodes $x$  to $\phi'(x)$ or $\phi''(x)$ depending whether it is in $P'$ or $P''$.  It is easy to verify that $\phi$ is a desired mapping.
\end{proof}

\end{document}